\documentclass{llncs}

 \usepackage{microtype}
 \usepackage{ltl}
 \usepackage{csquotes}
 \usepackage{amsmath}
 \usepackage{xcolor}
 \usepackage{mathtools}

 \newcommand\comment[1]{}
 \newcommand\varset[1]{\vec{#1}}
 
 \newcommand\chahn[1]{\textcolor{blue}{C. Hahn: #1}}
 \newcommand{\stone}[2]{\smash{\begin{tikzpicture}[baseline={(0,-0.02)}] \node at (0,0.15) {\tiny #1};\node at (0,0) {\tiny #2};\end{tikzpicture}}}
 \newcommand{\all}{$\forall^*$ }
 \newcommand{\ex}{$\exists^*$ }
 \newcommand{\allex}{$\forall^*\exists^*$ }
 \newcommand{\exall}{$\exists^*\forall^*$ }
 \definecolor{yell}{rgb}{0.99,0.78,0.07}
 \definecolor{lgreen}{RGB}{144,238,144}
 \definecolor{tred}{RGB}{255,99,71}
\begin{document}
 \title{Deciding Hyperproperties\thanks{This work was partially supported by the German Research Foundation (DFG) in the Collaborative Research Center 1223 and by the Graduate School of Computer Science at Saarland University.}}
 \author{Bernd Finkbeiner and Christopher Hahn}
 \institute{Saarland University \\ Saarbr\"ucken, Germany \\
 	\email{\{finkbeiner|hahn\}@react.uni-saarland.de}}
 	\maketitle
 	
 	\begin{abstract}
 		Hyperproperties, like observational determinism or symmetry, cannot be expressed as properties of individual computation traces, because they describe a relation between multiple computation traces. HyperLTL is a temporal logic that captures such relations through trace variables, which are introduced through existential and universal trace quantifiers and can be used to refer to multiple computations at the same time. In this paper, we study the satisfiability problem of HyperLTL. We show that the problem is PSPACE-complete for alternation-free formulas (and, hence, no more expensive than LTL satisfiability), EXPSPACE-complete for $\exists^*\forall^*$ formulas, and undecidable for $\forall\exists$ formulas. 
 		Many practical hyperproperties can be expressed as alternation-free formulas. Our results show that both satisfiability and implication are decidable for such properties.  
 	\end{abstract}
 	
 	\section{Introduction}
 	
 	Hyperproperties~\cite{DBLP:journals/jcs/ClarksonS10} are system properties that relate multiple computation traces. For example, in the design of a system that handles sensitive information, we might specify that a certain secret is kept confidential by requiring that the
 	system is \emph{deterministic} in its legitimately observable inputs, i.e., that
 	all computations with the same observable inputs must have the same observable outputs, independently of the secret~\cite{DBLP:conf/sp/Roscoe95,DBLP:conf/csfw/ZdancewicM03}. In the design of an access protocol for a shared resource, we might specify that the access to the resource is \emph{symmetric} between multiple clients by requiring that for every computation and every permutation of the clients, there exists a computation where the access is granted in the permuted order~\cite{DBLP:conf/cav/FinkbeinerRS15}.
 	
 	To express hyperproperties in a temporal logic, linear-time temporal logic (LTL) has recently been extended with trace variables and trace quantifiers. In HyperLTL~\cite{DBLP:conf/post/ClarksonFKMRS14}, observational determinism can, for example, be expressed as the formula
 	$\forall\pi.\forall\pi'.\;\LTLsquare(I_\pi=I_{\pi'})\rightarrow\LTLsquare(O_\pi=O_{\pi'}),$
 	where $I$ is the set of observable inputs and $O$ is the set of observable outputs. The universal quantification of the trace variables $\pi$ and $\pi'$ indicates that the property must hold for all pairs of computation traces. It has been shown that many hyperproperties of interest can be expressed in HyperLTL~\cite{markusPhD}.
 	
 	In this paper, we study the \emph{satisfiability problem} of HyperLTL.
 	Unlike the model checking problem, for which algorithms and tools
 	exist~\cite{DBLP:conf/post/ClarksonFKMRS14,DBLP:conf/cav/FinkbeinerRS15}, the decidability and complexity of the satisfiability problem
 	was, so far, left open. The practical demand for a decision procedure is
 	strong. Often, one considers multiple formalizations of similar, but
 	not necessarily equivalent hyperproperties. An alternative (and slightly stronger)
 	version of observational determinism requires, for example, that
 	differences in the observable output may only occur \emph{after} differences
 	in the observable input have occurred:
 	$\forall\pi.\forall\pi'.\;(O_\pi=O_{\pi'}) ~{\mathcal W}~
 	(I_\pi\neq I_{\pi'})$. A decision procedure for HyperLTL would allow
 	us to automatically check whether such formalizations imply each other. Another important application is to check whether the
 	functionality of a system, i.e., a standard trace property, is
 	compatible with the desired hyperproperties, such as confidentiality.
 	Since both types of properties can be expressed in HyperLTL, a decision procedure for HyperLTL 
 	would make it possible to identify inconsistent system requirements early on,
 	before an attempt is made to implement the requirements.
 	
 	The fundamental challenge in deciding hyperproperties is that 
 	hyperproperties are usually not
 	$\omega$-regular~\cite{DBLP:conf/icalp/AlurCZ06}.
 	HyperLTL formulas thus cannot be translated into equivalent
 	automata~\cite{DBLP:journals/it/FinkbeinerR14}. Intuitively, since
 	hyperproperties relate multiple infinite traces, an automaton, which only considers one trace at a time, would have to memorize an infinite amount of information from one trace to the next. This means that the
 	standard recipe for checking the satisfiability of a temporal logic,
 	which is to translate the given formula into an equivalent B\"uchi
 	automaton and then check if the language of the automaton is
 	empty~\cite{DBLP:journals/iandc/VardiW94}, cannot be applied to
 	HyperLTL.
 	
 	In model checking, this problem is sidestepped by verifying the
 	\emph{self-composition}~\cite{DBLP:conf/csfw/BartheDR04} of the given system: instead of verifying a
 	hyperproperty that refers to $n$ traces, we
 	verify a trace property that refers to a \emph{single} trace of a new
 	system that contains $n$ copies of the original system.
 	Since the satisfiability problem does not refer to a system, this idea cannot
 	immediately be applied to obtain a decision procedure for HyperLTL. However, it would seem natural to define a similar self-composition, on the formula rather than the system, in order to determine satisfiability.

 	
 	We organize our investigation according to the quantifier structure of
 	the HyperLTL formulas. LTL, for which the satisfiablity problem is already solved~\cite{DBLP:journals/jacm/SistlaC85}, is the sublogic of HyperLTL where the formulas have a single universally quantified trace variable, which is usually left implicit. The next larger fragment consists of the alternation-free formulas, i.e., formulas with an arbitrary number of trace variables and a quantifier prefix that either consists of only universal or only existential quantifiers. Many hyperproperties of practical interest, such as observational determinism, belong to this fragment.
 	It turns out that the satisfiability of alternation-free formulas can indeed be reduced to the satisfiability of LTL formulas by replicating the atomic propositions such that there is a separate copy for each trace variable. This construction is sound, because in an alternation-free formula, the values for the quantifiers can be chosen independently of each other. The size of the resulting LTL formula is the same as the given HyperLTL formula; as a result, the satisfiability problem of the alternation-free fragment has the same complexity, PSPACE-complete, as LTL satisfiability.
 	
 	If the formula contains a quantifier alternation, the values of the quantifiers can no longer be chosen independently of each other. However, if
 	the quantifier structure is of the form $\exists^*\forall^*$, i.e., the formula begins with an existential quantifier and then has a single quantifier alternation, then it is still possible to reduce HyperLTL satisfiability to LTL satisfiability by explicitly considering all possible interactions between the existential and universal quantifiers. For example, $\exists \pi_0 \exists \pi_1 \forall \pi_{2} .~ (\LTLcircle p_{\pi_0}) \wedge (\LTLsquare p_{\pi_1}) \wedge (\LTLdiamond p_{\pi_2})$ is equisatisfiable to $\exists \pi_0 \exists \pi_1 .~ (\LTLcircle p_{\pi_0}) \wedge (\LTLsquare p_{\pi_1}) \wedge (\LTLdiamond p_{\pi_0})  \wedge (\LTLdiamond p_{\pi_1})$, which is in turn equisatisfiable to the LTL formula $(\LTLcircle p_0) \wedge (\LTLsquare p_1) \wedge (\LTLdiamond p_0)  \wedge (\LTLdiamond p_1)$. In general, enumerating all combinations of existential and universal quantifiers causes an exponential blow-up and we show that the satisfiability problem for the  $\exists^*\forall^*$-fragment is indeed EXPSPACE-complete. This high complexity is, however, relativized by the fact that practical hyperproperties rarely need a large number of quantifiers. If we bound the number of universal quantifiers by a constant, the complexity becomes PSPACE again.
 	
 	Formulas where an existential quantifier occurs in the scope of a universal quantifier make the logic dramatically more powerful, because they can be used to enforce, inductively, models with an infinite number of traces. We show that a single pair of quantifers of the form $\forall \exists$ suffices to encode Post's correspondence problem. The complete picture is thus as summarized in Table~\ref{summary}: The largest decidable fragment of HyperLTL is the EXPSPACE-complete $\exists^*\forall^*$ fragment. Bounding the number of universal quantifiers and in particular restricting to alternation-free formulas reduces the complexity to PSPACE. Any fragment that contains the $\forall\exists$ formulas is undecidable.
 	
 	From a theoretical point of view, the undecidability of the $\forall\exists$ fragment is a noteworthy result, because it confirms the intuition that hyperproperties are truly more powerful than trace properties. In practice, already the alternation-free fragment suffices for many important applications (cf.~\cite{DBLP:conf/cav/FinkbeinerRS15}). From a practical point of view, the key result of the paper is therefore that both satisfiability of alternation-free formulas and implication between alternation-free formulas, which can be expressed as an unsatisfiability check of an \exall formula, are decidable.
 	
 	\begin{table}[t]
 		\centering
 		\caption{Complexity results for the satisfiability problem of HyperLTL}
 		\label{summary}
 		\begin{tabular}{|l|l|l|l|l|}
 			\hline
 			{\centering \ex}              & \all         & \exall  & \begin{tabular}[c]{@{}l@{}} bounded \\ \exall \end{tabular} & $\forall\exists$ \\ \hline
 			\begin{tabular}[c]{@{}l@{}}{PSPACE-}\\ {complete}\end{tabular} & \begin{tabular}[c]{@{}l@{}}{PSPACE-}\\ {complete}\end{tabular} & \begin{tabular}[c]{@{}l@{}}{EXPSPACE-}\\ {complete}\end{tabular} &\begin{tabular}[c]{@{}l@{}}{PSPACE-}\\ {complete}\end{tabular} & {undecidable} \\ \hline
 		\end{tabular}
 	\end{table}

 	\comment{
 		The reminder of the paper is structured as follows.
 		We define HyperLTL in Section~\ref{tracevshyper}.
 		In Section~\ref{alt-free-frag}, we show that the satisfiability problems of the alternation-free fragment of HyperLTL is PSPACE-complete. In Section~\ref{exall}, we provide an EXPSPACE-algorithm for deciding the satisfiability of the fragment with one quantifier alternation starting with an existential quantifier. We establish EXPSPACE-completeness via an encoding of an exponential space-bounded Turing machine into the \exall fragment of HyperLTL. We also study a bounded version of the satisfiability problem, which remains PSPACE-complete for this fragment.
 		As an application, we give an EXPSPACE-algorithm for deciding whether two universally quantified HyperLTL formulae are equivalent in Section~\ref{exall}.
 		We conclude the paper with the proof that the satisfiability problem remains undecidable for every other fragment of HyperLTL in Section~\ref{allex}.
 	} 
 	

 	\section{HyperLTL}
 	\label{tracevshyper}
 	Let $\mathit{AP}$ be a set of \emph{atomic propositions}.
 	A \emph{trace} $t$ is an infinite sequence over subsets of the atomic propositions. We define the set of traces $\mathit{TR} \coloneqq (2^\mathit{AP})^\omega$.
 	A subset $T \subseteq \mathit{TR}$ is called a \emph{trace property}.
 	We use the following notation to manipulate traces:
 	let $t \in \mathit{TR}$ be a trace and $i \in \mathbb{N}$ be a natural number. $t[i]$ denotes the $i$-th element of $t$. Therefore, $t[0]$ represents the starting element of the trace. Let $j \in \mathbb{N}$ and $j \geq i$. $t[i,j]$ denotes the sequence $t[i]~t[i+1]\ldots t[j-1]~t[j]$. $t[i, \infty]$ denotes the infinite suffix of $t$ starting at position $i$.
 	
 	\paragraph{LTL Syntax.}
 	Linear-time temporal logic (LTL)~\cite{DBLP:conf/focs/Pnueli77} combines the usual boolean connectives with temporal modalities such as the \emph{Next} operator $\LTLnext$ and the \emph{Until} operator $\LTLuntil$. The syntax of LTL is given by the following grammar:
 	
 	\begin{align*}
 	\varphi~& \Coloneqq~p~~|~~\neg \varphi~~|~~\varphi \vee \varphi~~|~~\LTLnext \varphi~~|~~\varphi\, \LTLuntil \varphi
 	\end{align*}
 	where $p \in \mathit{AP}$ is an atomic proposition.
 	$\LTLnext \varphi$ means that $\varphi$ holds in the \emph{next} position of a trace; $\varphi_1 \LTLuntil \varphi_2$ means that $\varphi_1$ holds \emph{until} $\varphi_2$ holds.
 	There are several derived operators, such as
 	$\LTLdiamond \varphi \equiv \mathit{true} \LTLuntil \varphi$,
 	$\LTLsquare \varphi \equiv \neg \LTLdiamond \neg \varphi$, and $\varphi_1\,\mathcal\, W \varphi_2 \equiv (\varphi_1 \LTLuntil \varphi_2) \vee \LTLsquare \varphi_1$.  
 	$\LTLdiamond \varphi$ states that $\varphi$ will \emph{eventually} hold in the future and $\LTLsquare$ states that $\varphi$ holds \emph{globally}; $\mathcal W$ is the \emph{weak} version of the \emph{until} operator. 
 	
 	\paragraph{LTL Semantics.} Let $p \in \mathit{AP}$ and $t \in \mathit{TR}$.
 	The semantics of an LTL formula is defined as the smallest relation $\models$ that satisfies the following conditions:
 	\begin{align*}
 	& t \models p &&~\text{iff} \hspace{5ex} p \in t[0] \\
 	& t \models \neg \psi &&~\text{iff} \hspace{5ex} t \not \models \psi \\
 	& t \models \psi_1 \vee \psi_2 &&~\text{iff} \hspace{5ex} t \models \psi_1~\text{or}~t \models \psi_2 \\
 	& t \models \LTLnext \psi &&~\text{iff} \hspace{5ex} t [1,\infty] \models \psi \\
 	& t \models \psi_1 \LTLuntil \psi_2 &&~\text{iff}\hspace{5ex} \text{there exists}~i \geq 0 : t[i,\infty] \models \psi_2 \\
 	& &&\hspace{7.8ex} \text{and for all}~0 \leq j < i~\text{we have}~t[j,\infty] \models \psi_1
 	\end{align*}
 	\emph{LTL-SAT} is the problem of deciding whether there exists a trace $t \in \mathit{TR}$ such that $t \models \psi$.
 	\begin{theorem}
 		\label{t1}
 		LTL-SAT is PSPACE-complete~\cite{DBLP:journals/jacm/SistlaC85}.
 	\end{theorem}
 	
 	
 	
 	\paragraph{HyperLTL Syntax.}
 	HyperLTL~\cite{DBLP:conf/post/ClarksonFKMRS14} extends LTL with trace variables and trace quantifiers. Let $\mathcal{V}$ be an infinite supply of trace variables.  The syntax of HyperLTL is given by the following grammar:
 	\begin{align*}
 	\psi~&\Coloneqq~\exists \pi.\;\psi~~|~~\forall\pi.\;\psi~~|~~\varphi \\
 	\varphi~&\Coloneqq~a_{\pi}~~|~~\neg \varphi~~|~~\varphi \vee \varphi~~|~~\LTLnext \varphi~~|~~\varphi\, \LTLuntil \varphi 
 	\end{align*}
 	where $a\in \mathit{AP}$ is an atomic proposition and $\pi \in \mathcal V$ is a trace variable. Note that atomic propositions are indexed by trace variables.
 	The quantification over traces makes it possible to express properties like ``on all traces $\psi$ must hold'', which is expressed by $\forall \pi.~\psi$. Dually, one can express that ``there exists a trace such that $\psi$ holds'', which is denoted by $\exists \pi.~\psi$. The derived operators $\LTLdiamond$, $\LTLsquare$, and $\mathcal W$ are defined as for LTL.
 	
 	\paragraph{HyperLTL Semantics.}
 	A HyperLTL formula defines a \emph{hyperproperty}, i.e., a set of sets of traces. A set $T$ of traces satisfies the hyperproperty if it is an element of this set of sets. 
 	Formally, the semantics of HyperLTL formulas is given with respect to a \emph{trace assignment} $\Pi$ from $\mathcal{V}$ to $\mathit{TR}$, i.e. a partial function mapping trace variables to actual traces. $\Pi[\pi \mapsto t]$ denotes that $\pi$ is mapped to $t$, with everything else mapped according to $\Pi$. $\Pi[i,\infty]$ denotes the trace assignment that is equal to $\Pi(\pi)[i,\infty]$ for all $\pi$.
 	\begin{align*}
 	&\Pi \models_T~\exists \pi. \psi &&\text{iff}\hspace{5ex} \text{there exists}~t \in T~:~ \Pi[\pi \mapsto t] \models_T \psi \\
 	&\Pi \models_T~\forall \pi. \psi &&\text{iff}\hspace{5ex} \text{for all}~t \in T~:~ \Pi[\pi \mapsto t] \models_T \psi \\
 	&\Pi \models_T~a_{\pi} &&\text{iff}\hspace{5ex} a \in \Pi(\pi)[0] \\
 	&\Pi \models_T~\neg \psi &&\text{iff}\hspace{5ex} \Pi \not \models_T \psi \\
 	&\Pi \models_T~\psi_1 \vee \psi_2 &&\text{iff}\hspace{5ex} \Pi \models_T \psi_1~\text{or}~\Pi \models_T \psi_2 \\
 	&\Pi \models_T~\LTLnext \psi &&\text{iff}\hspace{5ex} \Pi[1,\infty] \models_T \psi \\
 	&\Pi \models_T~\psi_1 \LTLuntil \psi_2 &&\text{iff}\hspace{5ex} \text{there exists}~i \geq 0 : \Pi[i,\infty] \models_T \psi_2 \\
 	& &&\hspace{7.2ex} \text{and for all}~0 \leq j < i~\text{we have}~\Pi[j,\infty] \models_T \psi_1
 	\end{align*}
 	\emph{HyperLTL-SAT} is the problem of deciding whether there exists a \emph{non-empty} set of traces $T$ such that $\Pi \models_T \psi$, where $\Pi$ is the empty trace assignment and $\models_T$ is the smallest relation satisfying the conditions above.
 	If it is clear from the context, we omit $\Pi$ and simply write $T \models \psi$.
 	If $\models_T \psi$, we call $T$ a \emph{model} of $\psi$.
 	
 	
 	\section{Alternation-free HyperLTL}
 	\label{alt-free-frag}
 	We begin with the satisfiability problem for the alternation-free fragments of HyperLTL.
 	We call a HyperLTL formula $\psi$ (quantifier) \emph{alternation-free} iff the quantifier prefix only consists of  either only universal or only existential quantifiers. We denote the corresponding fragments as the $\forall^*$ and $\exists^*$ fragments, respectively. For both fragments, we show that every formula can be reduced, as discussed in the introduction, to an equisatisfiable LTL formula of the same size. As a result, we obtain that the satisfiability problem of alternation-free HyperLTL is PSPACE-complete, like the satisfiability problem of LTL.

 	\subsection{The \boldmath{$\forall^*$} Fragment}
 	
 	The \all fragment is particularly easy to decide, because we can restrict the models, without loss of generality, to singleton sets of traces: since all quantifiers are universal, every model with more than one trace could immediately be translated into another one where every trace except one is omitted.
 	Hence, we can ignore the trace variables and interpret the HyperLTL formula as a plain LTL formula.
 	
 	\begin{example}
 		Consider the following HyperLTL formula with atomic propositions $\{a,b\}$:
 		\begin{align*}
 		\forall \pi_1 \forall \pi_2.\;\LTLsquare b_{\pi_1} \wedge \LTLsquare \neg b_{\pi_2}
 		\end{align*}
 		Since the trace variables are universally quantified, we are reasoning about \emph{every} pair of traces, and thus in particular about the pairs where both variables refer to the same trace. It is, therefore, sufficient to check the satisfiability of the LTL formula $\LTLsquare b \wedge \LTLsquare \neg b$, which turns out to be unsatisfiable.
 	\end{example}
 	
 	The satisfiability of hyperproperties that can be expressed in the \all fragment, such as observational determinism, thus immediately reduces to LTL satisfiability.
 	
 	\begin{definition}
 		\label{phi_l}
 		Given a HyperLTL formula $Q.\varphi$, where $Q$ is an arbitrary quantifier prefix. We delete every trace variable in $\varphi$ and omit the quantifier prefix $Q$. We denote the resulting LTL formula by $\varphi^{-l}$.
 	\end{definition}
 	
 	We prove that the previous definition preserves satisfiability with the following lemma. It states that, for a given \all HyperLTL formula, every trace of a model does also satisfy the LTL formula obtained by Definition~\ref{phi_l}.
 	
 	\begin{lemma}
 		\label{choice}
 		For every \all HyperLTL formula, the following holds:
 		\begin{align*}
 		\forall~T \in 2^\mathit{TR}. ((T \models \forall \pi_1 \ldots \forall \pi_n.~\varphi) \rightarrow \forall t \in T.~ t \models \varphi^{-l})
 		\end{align*}
 	\end{lemma}
 	
 	\begin{proof}
 		Let a \all HyperLTL formula $\forall \pi_1 \ldots \forall \pi_n.~\varphi$ over a set of atomic propositions $\mathit{AP}$ be given. Let $T \in 2^\mathit{TR}$ be an arbitrary trace set satisfying $\forall \pi_1 \ldots \pi_n.~\varphi$. We distinguish two cases. If $T$ is the empty set, then the statement trivially holds.
 		Otherwise, let $t \in T$ be arbitrary. By assumption the singleton $\{t\}$ satisfies $\forall \pi_1 \ldots \forall \pi_n.~\varphi$ by assigning $t$ to every trace variable $\pi_1, \ldots, \pi_n$. Hence $t$ satisfies $\varphi^{-l}$. \qed
 	\end{proof}
 	
 	\begin{lemma}
 		\label{lemmaall}
 		For every \all HyperLTL formula there exists an equisatisfiable LTL formula of the same size.
 	\end{lemma}
 	
 	\begin{proof}
 		Let a \all HyperLTL formula $\forall \pi_1 \ldots \forall \pi_n.~\varphi$ and its corresponding LTL formula $\varphi^{-l}$ be given.
 		\begin{itemize}
 			\item
 			Assume $\forall \pi_1 \ldots \forall \pi_n.~\varphi$ is satisfiable by a trace set $T$. Choose an arbitrary trace $t \in T$. By Lemma~\ref{choice}, $t$ is a witness for the satisfiability of $\varphi^{-l}$, i.e., $\varphi^{-l}$ is satisfiable.
 			\item
 			Assume the LTL formula $\varphi^{-l}$ is satisfiable by a trace $t$. A trace set satisfying \linebreak $\forall \pi_1 \ldots \forall \pi_n.~\varphi$ is $\{t\}$, since $t$ is assigned to every universally quantified trace variable.
 		\end{itemize}
 		Therefore a \all HyperLTL formula $\forall \pi_1 \ldots \forall \pi_n.~\varphi$ is equisatisfiable to its corresponding LTL formula $\varphi^{-l}$. \qed
 	\end{proof}
 	
 	With this lemma, we provide proof that HyperLTL-SAT of \all formulas inherits the complexity of LTL-SAT, which is PSPACE-complete.
 	
 	\begin{lemma}
 		\label{complexity_forall_n}
 		HyperLTL-SAT is {PSPACE-complete} for the \all fragment.
 	\end{lemma}
 	
 	\begin{proof}
 		We use Definition~\ref{phi_l} to transform an arbitrary \all HyperLTL formula to an LTL formula in polynomial space, which is, by Lemma~\ref{lemmaall}, equisatisfiable to the original HyperLTL formula. LTL-SAT is, by Theorem~\ref{t1}, in {PSPACE} and thus \all HyperLTL-SAT is in {PSPACE} too.\\
 		We prove hardness with a reduction from LTL-SAT, which is, by Theorem~\ref{t1}, PSPACE-hard. An LTL formula $\varphi$ is reduced to a HyperLTL formula $\forall \pi.~\varphi(\pi)$, where $\varphi(\pi)$ denotes that every atomic proposition occurring in $\varphi$ is labelled with $\pi$.
 		The resulting formula $\forall \pi.~\varphi(\pi)$ is satisfiable by the singleton $\{t\}$ if and only if $\varphi$ is satisfiable by $t$, for all $t \in \mathit{TR}$.
 		Hence, \all HyperLTL-SAT is {PSPACE-complete}. \qed
 	\end{proof}
 	
 	\subsection{The \boldmath{$\exists^*$} Fragment}
 	\label{secex}
 	
 	A model of a formula in the \ex fragment may, in general, have more than one trace. For example the models of $\exists \pi_1 \exists \pi_2.~ a_{\pi_1} \wedge \neg a_{\pi_2}$ have (at least) two traces. In order to reduce HyperLTL satisfiability again to LTL satisfiability, we \emph{zip} such traces together.
 	For this purpose, we introduce a fresh atomic proposition for every atomic proposition $a$ and every path variable $\pi$ that occur as an indexed proposition $a_\pi$ in the formula. We obtain an equisatisfiable LTL formula by removing the quantifier prefix and replacing every occurrence of $a_\pi$ with the new proposition.
 	
 	\begin{example}
 		Consider the following HyperLTL formula over the atomic propositions ${\{a,b\}}$:
 		\begin{align*}
 		\exists \pi_1 \exists \pi_2.&~a_{\pi_1} \wedge \LTLsquare \neg b_{\pi_1} \wedge \LTLsquare b_{\pi_2}
 		\end{align*}
 		By discarding the quantifier prefix and replacing the indexed propositions with fresh propositions, we obtain the equisatisfiable LTL formula over the atomic propositions $\{a_1, b_1,b_2\}$:
 		\begin{align*}
 		a_1 \wedge \LTLsquare \neg b_1 \wedge \LTLsquare b_2
 		\end{align*}
 		The LTL formula is satisfied by the trace $\tilde{p}$: $(\{a_1,b_2\})^\omega$.
 		We can map the fresh propositions back to the original indexed propositions. In this way, we obtain witnesses for $\pi_1$ and $\pi_2$ by splitting $\tilde{p}$ into two traces $\{a\}^\omega$ and $\{b\}^\omega$, where for every position in these traces only those atomic propositions that were labelled with $\pi_1$ or $\pi_2$, respectively, hold. Hence, the trace set satisfying the HyperLTL formula is $\{\{a\}^\omega,\{b\}^\omega\}$.
 	\end{example}
 	
 	We formally define a mapping that replaces atomic propositions, which are labelled with trace variables, with fresh atomic propositions, that contain the information, which is lost by discarding the quantifier prefix. For example, $a_{\pi_m}$ will be replaced with $a_m$.
 	\begin{definition}
 		\label{psi_exists}
 		Let $\mathit{AP}$ be $\{a_1,\ldots, a_m\}$ and $\exists \pi_1 \ldots \exists \pi_n.~\varphi$ be an \ex HyperLTL formula.
 		We create an LTL formula $\varphi_\exists$ by replacing every trace-labelled atomic proposition $a_{1_{\pi_1}},\ldots,a_{1_{\pi_n}},\ldots,a_{m_{\pi_1}},\ldots,a_{m_{\pi_n}}$ with fresh atomic propositions of $\widetilde{\mathit{AP}}\coloneqq \bigcup^{m}_{j=1} \bigcup^{n}_{i=1} a_{j_i}$. Furthermore, we discard the quantifier prefix.
 		The substitution is represented by the mapping $s \in \mathit{AP} \times \{1,\ldots,n\} \rightarrow \widetilde{\mathit{AP}}$, where $\{1,\dots,n\}$ represents the trace variables. We will use $\varphi_\exists$ as notation for applying this construction.
 	\end{definition}
 	
 	After the replacement, we can construct a trace set by projecting every atomic proposition back to its corresponding trace, since
 	Definition~\ref{psi_exists} preserves the necessary information.
 	For example, consider a trace where $a_1$ and $a_{2}$ holds on the first position. The atomic proposition $a \in \mathit{AP}$ will be put in the first position of the first and the second trace, but no other.
 	
 	\begin{definition}
 		\label{pr}
 		Let $\mathit{AP}$ be a set of atomic propositions. Let $s \in \mathit{AP} \times \{1,\ldots, n\} \rightarrow \widetilde{\mathit{AP}}$ be a substitution that maps trace-labelled atomic propositions to fresh ones. Let $\tilde{t}$ be a trace that ranges over $\widetilde{\mathit{AP}}$ and $\varphi$ an \ex HyperLTL formula with $n$ quantifier. We also define $s(M,i)$ for a set $M \subseteq \mathit{AP}$ and $1 \leq i \leq n$ as $\{s(p,i)~|~ p \in M\}$. The projection $pr^s_i$ for the $i$th trace variable is defined as the trace $t$, s.t. $\forall j \in \mathbb{N}. t[j] = s(\tilde{t}[j],i)^{-1}$. This denotes that only the atomic propositions that correspond to the $i$th trace variable are kept in the corresponding trace. Therefore, $pr^s(\tilde{t}) :=$ $\bigcup_{i=1}^{n }~pr^s_i(\tilde{t})$ splits the trace $t$ into a trace set containing $n$ traces, each being a witness for its respective trace quantifier $\pi_i$.
 	\end{definition}
 	
 	\begin{lemma}
 		\label{lemmaex}
 		For every \ex HyperLTL formula there exists an equisatisfiable LTL formula of the same size.
 	\end{lemma}
 	
 	\begin{proof}Let $AP$ be a set of atomic propositions. Let $s \in AP \times \{1,\ldots,n\} \rightarrow \widetilde{\mathit{AP}}$ be as defined in Definition~\ref{psi_exists}.
 		\begin{itemize}
 			\item
 			Let $\exists \pi_1 \ldots \exists \pi_n.~\varphi$ be an \ex HyperLTL formula. Assume there exists a model $T \subseteq \mathit{TR}$, s.t. $T \models \exists \pi_1, \ldots, \pi_n.~\varphi$. 
 			Let $t_1,\ldots,t_n \in T$ be the witnesses of $\pi_1,\ldots, \pi_n$.
 			We define $\tilde{t}$ as the trace that is generated by zipping the traces $t_1,\ldots,t_n \in T$ together, which means that $\forall j \in \mathbb{N}.~\tilde{t}[j]:= \bigcup^n_{i=1} s(t_i[j],i)$. Since $\tilde{t}$ is exactly the trace that satisfies $\varphi_\exists$ by Definition~\ref{psi_exists}, $\tilde{t} \models \varphi_\exists$ must hold. This means we found a witness for satisfiability, namely $\tilde{t}$.
 			\item
 			Let $\exists \pi_1 \ldots \exists \pi_n.~\varphi$ be an \ex HyperLTL formula and $\varphi_\exists$ be the corresponding LTL formula from Definition~\ref{psi_exists}.
 			Assume $\tilde{t}$ satisfies $\varphi_\exists$. We show that we can find witnesses $t_1,\ldots,t_n$, s.t. $\{t_1,\ldots,t_n\}\models \exists \pi_1 \ldots \exists \pi_n.~\varphi$.
 			We use Definition~\ref{pr} to construct the desired trace set $\{t_1,\ldots,t_n\}$, which is $pr^s(\tilde{t})$.
 			Let $k$ be an arbitrary position in $\tilde{t}$, $i$ be an arbitrary quantifier index, and $a$ be an arbitrary atomic proposition, where $s(a,i) = a'$ and, therefore, $a' \in \widetilde{\mathit{AP}}$.
 			We distinguish two cases.
 			\begin{itemize}
 				\item
 				$a' \in \tilde{t}[k]$:
 				Assume $a'$ holds at position $k$ in $\tilde{t}$. The projection $pr^s$ simply projects $a'$ to $a$ at position $k$ of the witness of the $i$th quantifier, i.e., $a \in t_i[k]$.
 				\item
 				$a' \not\in \tilde{t}[k]$:
 				If $a'$ does not hold at position $k$, $a'$ is not projected to position $k$ in $p_i$, i.e., $a \not\in t_i[k]$.
 			\end{itemize}
 		\end{itemize}
 		Satisfiability is, therefore, preserved by the projection $pr^s$ of Definition~\ref{pr}.
 		With Definition~\ref{psi_exists} above we have shown that $\exists \pi_1 \ldots \exists \pi_n.~\varphi$ and $\varphi_\exists$ are equisatisfiable. \qed
 	\end{proof}
 	
 	Using Lemma~\ref{lemmaex}, we prove the following lemma and conclude with Lemma~\ref{lemmaall} that HyperLTL-SAT inherits the complexity of LTL-SAT for the alternation-free fragment.
 	
 	\begin{lemma}
 		\label{complexity_exists_n}
 		HyperLTL-SAT is {PSPACE-complete} for the \ex fragment.
 	\end{lemma}
 	\begin{proof}
 		We use Definition~\ref{psi_exists} to transform an arbitrary \ex HyperLTL formula to a plain LTL formula using polynomial space, which is, by Lemma~\ref{lemmaex}, equisatisfiable to the original HyperLTL formula. LTL-SAT is, by Theorem~\ref{t1}, in PSPACE and thus \ex HyperLTL-SAT is in {PSPACE} too.\\
 		We prove hardness with a reduction from LTL-SAT, which is, by Theorem~\ref{t1}, PSPACE-hard. An LTL formula $\varphi$ is reduced to a HyperLTL formula $\exists \pi.~\varphi(\pi)$, where $\varphi(\pi)$ denotes that every atomic proposition occurring in $\varphi$ is labelled with $\pi$.
 		By definition the resulting formula $\exists \pi.~\varphi(\pi)$ is satisfiable if and only if $\varphi$  is satisfiable.
 		Hence, \ex HyperLTL-SAT is {PSPACE-complete}. \qed
 	\end{proof}
 	
 	\begin{theorem}
 		\label{alt-free-complexity}
 		HyperLTL-SAT is PSPACE-complete for the alternation-free fragment.
 	\end{theorem}
 	
 	\begin{proof}
 		Follows directly from Lemma~\ref{complexity_forall_n} and Lemma~\ref{complexity_exists_n}. \qed
 	\end{proof}
 	
 	\section{The \boldmath{$\exists^*\forall^*$} Fragment}
 	\label{exall}
 	Allowing quantifier alternation makes the satisfiability problem significantly more difficult, and even leads to undecidability, as we will see in the next section. In this section, we show that deciding formulas with a single quantifier alternation is still possible if the quantifiers start with an existential quantifier.
 	A HyperLTL formula is in the \exall fragment iff it is of the form $\exists \pi_1 \ldots \exists \pi_n \forall \pi'_1 \ldots \forall \pi'_m.\;\psi$.
 	This fragment is especially interesting, because it includes implications between alternation-free formulas. 
 	The idea of the decision procedure is to eliminate the universal quantifiers by explicitly enumerating all possible interactions between the universal and existential quantifiers. This leads to an exponentially larger, but equisatisfiable \ex formula.
 	
 	\begin{lemma}
 		\label{sp}
 		For every formula in the \exall fragment, there is an equisatisfiable formula in the \ex fragment with exponential size.
 	\end{lemma}
 	
 	\begin{proof}
 		We define a function $\mathit{sp}$ that takes a formula of the form $\exists \pi_1 \ldots \exists \pi_n \forall \pi'_1 \ldots \forall \pi'_m.~\psi$ and yields an \ex HyperLTL formula $\psi'$ of size $\mathcal{O}(n^m)$ of the following shape, where $\psi[\pi'_i \backslash \pi_i]$ denotes that the trace variable $\pi'_i$ in $\psi$ is replaced by $\pi_i$.
 		\begin{align*}
 		\exists \pi_1 \ldots \exists \pi_n.~\bigwedge_{j_1=1}^n~\ldots~\bigwedge_{j_m=1}^n.~\psi[\pi'_1\backslash \pi_{j_1}] \ldots \psi[\pi'_m\backslash \pi_{j_m}]
 		\end{align*}
 		Let $\varphi$ be an $\exists^* \forall^*$ HyperLTL formula satisfied by some model $T$. Hence, there exist traces $t_1, \ldots, t_n \in T$ such that $\{t_1, \ldots, t_n\}$ satisfies $\mathit{sp}(\varphi)$. Assume $\mathit{sp}(\varphi)$ is satisfied by some model $T'$. Since $\mathit{sp}$ covers every possible combination of trace assignments for the universally quantified trace variables, $T' \models \varphi$. \qed
 	\end{proof}

 	\begin{example}
 		Consider the \exall formula $\exists \pi_1 \exists \pi_2 \forall \pi'_{1} \forall \pi'_{2}.\;(\LTLsquare a_{\pi'_1} \wedge \LTLsquare b_{\pi'_2}) \wedge (\LTLsquare c_{\pi_1} \wedge \LTLsquare d_{\pi_2})$.
 		Applying the construction from Lemma~\ref{sp}, we obtain the following \ex formula:
 		\begin{align*}
 		sp(\exists \pi_1 \exists \pi_2 \forall \pi'_{1} \forall \pi'_{2}.\;&(\LTLsquare a_{\pi'_1} \wedge \LTLsquare b_{\pi'_2}) \wedge (\LTLsquare c_{\pi_1} \wedge \LTLsquare d_{\pi_2}))~\text{yields}: \\
 		\exists \pi_1 \exists \pi_2.\;((&\LTLsquare a_{\pi_1} \wedge \LTLsquare b_{\pi_1}) \wedge (\LTLsquare c_{\pi_1} \wedge \LTLsquare d_{\pi_2})) \\
 		{} \wedge ((&\LTLsquare a_{\pi_2} \wedge \LTLsquare b_{\pi_1}) \wedge (\LTLsquare c_{\pi_1} \wedge \LTLsquare d_{\pi_2}))	 \\
 		{} \wedge ((&\LTLsquare a_{\pi_1} \wedge \LTLsquare b_{\pi_2}) \wedge (\LTLsquare c_{\pi_1} \wedge \LTLsquare d_{\pi_2}))	 \\
 		{} \wedge ((&\LTLsquare a_{\pi_2} \wedge \LTLsquare b_{\pi_2}) \wedge (\LTLsquare c_{\pi_1} \wedge \LTLsquare d_{\pi_2}))
 		\end{align*}
 	\end{example}

 	Combining the construction from Lemma~\ref{sp} with the satisfiability check for \ex formulas from Section~\ref{alt-free-frag}, we obtain an exponential-space decision procedure for the \exall fragment.

 	\begin{theorem}
 		\label{theoremexall}
 		\exall HyperLTL-SAT is EXPSPACE-complete.
 	\end{theorem}
 	
 	\begin{proof}
 		Membership in EXPSPACE follows from Lemma~\ref{sp} and Lemma~\ref{lemmaex}.
 		We show EXPSPACE-hardness via a reduction from the problem whether an exponential-space bounded deterministic Turing machine $T$ accepts an input word $x$. Given $T$ and $x$, we construct an \exall HyperLTL formula $\varphi$ such that $T$ accepts $x$ iff $\varphi$ is satisfiable.
 		
 		Let $T=(\Sigma, Q, q_0, F, \rightarrow)$, where $\Sigma$ is the alphabet, $Q$ is the set of states, $q_0 \in Q$ is the initial state, $F\subseteq Q$ is the set of final states, and ${\rightarrow} \subseteq Q \times \Sigma \times Q \times \Sigma \times \{L,R\}$ is the transition relation. We use $(q,\sigma) \rightarrow (q',\sigma',\Delta)$ to indicate that when $T$ is in state $q$ and it reads the input $\sigma$ in the current tape cell, it changes its state to $q'$, writes $\sigma'$ in the current tape cell, and moves its head one cell to the left if $\Delta=L$ and one cell to the right if $\Delta = R$. Let $n \in \mathcal{O}(|x|)$ be such that the working tape of $T$ has $2^n$ cells.
 		We encode each letter of $\Sigma$ as a valuation of a set $\varset{s}=\{s_1, \ldots, s_{k_\Sigma}\}$ of atomic propositions and 
 		each state $Q$ as a valuation of another set $\varset{q}=\{q_1, \ldots, q_{k_Q}\}$ of atomic propositions, where ${k_\Sigma}$ is logarithmic in $|\Sigma|$ and $k_Q$ is logarithmic in $|Q|$. We furthermore use the valuations of a set $\varset{a}=\{a_1,\ldots, a_n\}$ to encode the position of a tape cell in a configuration of $T$, and the valuations of a set $\varset{h}=\{h_1,\ldots, h_n\}$ to encode the position of the head of the Turing machine.
 		With these atomic propositions, we can represent configurations of the Turing machine as sequences of valuations of the atomic propositions. The state of the Turing machine is encoded as the valuation of $\varset{q}$ at the position indicated by $\varset{h}$. Computations of a Turing machine are sequences of configurations; we thus represent computations as traces. 
 		
 		We begin our encoding into HyperLTL with four quantifier-free formulas over a free trace variable $\pi$:
 		$\varphi_{\mathit{init}}(\pi)$ encodes that the initial configuration represents $x$ and $q_0$, and places the head in the first position of the sequence. $\varphi_{\mathit{head}}(\pi)$ ensures that the position of the head may only change when a new configuration begins and that the change of the position as well as the change of the state is as defined by $\rightarrow$. $\varphi_{\mathit{count}}(\pi)$ expresses that the addresses in $\varset{a}$ continuously count from 1 to $2^n$. $\varphi_{\mathit{halt}}(\pi)$ expresses that the Turing machine halts eventually, i.e., the trace eventually visits a final state at the position of the head. 
 		
 		The more difficult part of the encoding now concerns the comparison of the tape content from one configuration to the next. We need to enforce that the tape content at the position represented by $\varset{h}$ changes as defined by $\rightarrow$, and that the content of all tape cells except for the position represented by $\varset{h}$ stays the same. For this purpose, we need to be able to memorize a position from one configuration to the next. We accomplish the ``memorization'' with the following trick: we introduce two existentially quantified trace variables $\pi_{\mathit{zero}}$ and $\pi_{\mathit{one}}$. Let $v$ be a new atomic proposition. We use a quantifier-free formula $\varphi_{\mathit{zero/one}}(\pi_{\mathit{zero}},\pi_{\mathit{one}})$ to ensure that $v$ is always $\mathit{false}$ on $\pi_{\mathit{zero}}$ and always $\mathit{true}$ on $\pi_{\mathit{one}}$. We now introduce another set of $n$ universally quantified trace variables $\pi_1, \pi_2, \ldots, \pi_n$ that will serve as memory: if one of these trace variables is bound to $\pi_{\mathit{zero}}$ its ``memory content'' is 0, if it is bound to $\pi_{\mathit{one}}$ its memory content is 1.
 		We add a sufficient number of universally quantified variables to memorize the position of some cell and its content.
 		Our complete encoding of the Turing machine as a HyperLTL formula then looks, so far, as follows:
 		\[ \begin{array}{l}
 		\exists \pi, \pi_{\mathit{zero}}, \pi_{\mathit{one}}.~ \forall \pi_1, \pi_2, \ldots, \pi_n, \pi'_1, \pi'_2, \ldots, \pi'_{k_\Sigma}.~  \\
 		\qquad  \varphi_{\mathit{init}}(\pi) \wedge \varphi_{\mathit{head}}(\pi) \wedge \varphi_{\mathit{count}}(\pi) \wedge \varphi_{\mathit{halt}}(\pi) \wedge\, \varphi_{\mathit{zero/one}}(\pi_{\mathit{zero}},\pi_{\mathit{one}})\\
 		\qquad \wedge\, \psi(\pi, \pi_1, \pi_2, \ldots, \pi_n, \pi'_1, \pi'_2, \ldots, \pi'_{k_\Sigma}) \end{array}\]
 		The missing requirement about the correct contents of the tape cells is encoded in the last conjunct $\psi$.
 		We first ensure that all the universally quantified traces have constant values in $v$, i.e., $v$ is either always $\mathit{true}$ or always $\mathit{false}$.
 		To enforce that the tape content changes at the head position, we specify in $\psi$ that whenever we are at the head position, i.e., whenever $a_{i,\pi} = h_{i,\pi}$ for all $i=1,\ldots,n$, then when we visit the same position in the next configuration, the tape content must be as specified by $\rightarrow$: i.e., if  $a_{i,\pi} = v_{\pi_i}$ for all $i=1,\ldots,n$, then when $a_{i,\pi} = v_{\pi_i}$ holds again for all $i=1,\ldots,n$ during the next configuration, the tape content as represented in $\varset{s}$ must be the one defined by $\rightarrow$. 
 		To enforce that the tape content is the same at every position except that encoded in $\varset{h}$, we specify that for all positions except the head position, i.e., whenever $a_{i,\pi} \neq h_{i,\pi}$ for some $i=1,\ldots,n$, then if  $a_{i,\pi} = v_{\pi_i}$ for all $i=1,\ldots,n$, and $s_{i,\pi}=v_{\pi'_i}$ for all $i=1,\ldots,k_\Sigma$, then the following must hold: when, during the next configuration, we visit the same position again, i.e., when again  $a_{i,\pi} = v_{\pi_i}$ for all $i=1,\ldots,n$,  we must also find the same tape content again, i.e., $s_{i,\pi}=v_{\pi'_i}$ for all $i=1,\ldots,k_\Sigma$.
 		
 		By induction on the length of the computation prefix, we obtain that any model of the HyperLTL formula represents in $\pi$ a correct computation of the Turing machine $T$. Since this computation must reach a final state, the model exists iff $T$ accepts the input word $x$. \qed
 	\end{proof}
 	
 	In practice, the number of quantifiers is usually small. Often it is sufficient to reason about pairs of traces, which can be done with just two quantifiers.
 	To reflect this observation, we define a \emph{bounded} version of the $\exists^* \forall^*$ fragment where the number of universal quantifiers that may occur in the HyperLTL formula is bounded by some constant $b \in \mathbb{N}$.
 	A bounded \exall formula of length $n$ with bound $b$ can be translated to an equisatisfiable LTL formulas of size $\mathcal{O}(n^b)$. The satisfiablility problem can thus be solved in polynomial space.
 	
 	\begin{corollary}
 		\label{complexity_bounded_exists_forall}
 		Bounded \exall HyperLTL-SAT is {PSPACE-complete}.
 	\end{corollary}
 	
 	Another observation that is important for the practical application of our results is that implication between alternation-free formulas is decidable. As discussed in the introduction, it frequently occurs that multiple formalizations are proposed for the same hyperproperty, and one would like to determine whether the proposals are equivalent, or whether one version is stronger than the other.
 	A HyperLTL formula $\psi$ \emph{implies} a HyperLTL formula $\varphi$ iff every set $T$ of traces that satisfies $\psi$ also satisfies $\varphi$.
 	
 	To determine whether $\psi$ implies $\varphi$, we check the satisfiability of the negation $\neg (\psi \rightarrow \varphi)$. If one formula is in the \all fragment and the other in the \ex fragment, implication checking is especially easy, because the formula we obtain is alternation-free.
 	If the given HyperLTL formulas are both in the $\forall^*$ fragment, respectively \ex fragment, we can construct an equisatisfiable \exall HyperLTL formula with the help of the following lemma.
 	
 	\begin{lemma}
 		\label{equi_forall}
 		Implication between two HyperLTL formulas of the \all fragment is decidable in exponential space.
 	\end{lemma}
 	
 	\begin{proof}
 		Let two HyperLTL formulas $\forall \pi_1 \ldots\forall \pi_n.~\psi$ and $\forall \pi_1'\ldots \forall \pi_m'.~\varphi$ be given. For determining if $\forall \pi_1 \ldots \forall \pi_n.~\psi$ implies $\forall \pi_1' \ldots \forall \pi_m'.~\varphi$, we will answer the equivalent question whether $\forall \pi_1 \ldots\forall \pi_n.~\psi \rightarrow \forall \pi_1'\ldots \forall \pi_m'.~\varphi$ is a tautology. To this end, we will check if the negation, $\neg (\forall \pi_1 \ldots \forall \pi_n.~\psi \rightarrow \forall \pi_1' \ldots \forall \pi_m'.~\varphi)$, is satisfiable.
 		Note that this is not (yet) a proper HyperLTL formula.
 		We can, however, simplify the formula and reduce the problem to satisfiability checking for plain LTL formulas. The following equivalence follows from basic logical definitions:
 		\[
 		\neg (((\forall \pi_1 \ldots \pi_n.~\psi) \rightarrow (\forall \pi_1'\ldots\pi_m'.~\varphi)) \equiv ((\forall \pi_1 \ldots \pi_n.~\psi) \wedge (\exists \pi_1' \ldots \pi_m'.~\neg \varphi))
 		\]
 		By quantifier rules, it is sufficient to check $(\exists \pi_1' \ldots \pi_m' \forall \pi_1 \ldots \pi_n.~\psi \wedge \neg \varphi)$ for satisfiability.
 		By Lemma~\ref{sp}, this formula is satisfiable iff $sp(\exists \pi_1' \ldots \pi_m' \forall \pi_1 \ldots \pi_n.~\psi \wedge \neg \varphi)$ is satisfiable.
 		By Lemma~\ref{lemmaex} and using the construction from Definition~\ref{psi_exists}, this formula is satisfiable iff the following formula is satisfiable:
 		\[\varphi_{\exists}(sp(\exists \pi_1' \ldots \pi_m' \forall \pi_1 \ldots \pi_n.~\psi \wedge \neg \varphi))\]
 		Note that this is a plain LTL formula.
 		If this formula is satisfiable, then the implication of the two HyperLTL formulas $\forall \pi_1\ldots\forall \pi_n.~\psi$ and $\forall \pi_1'\ldots\forall \pi_m'.~\varphi$ does not hold. Otherwise it does. \qed
 	\end{proof}
 	
 	\begin{lemma}
 		\label{equi_exists}
 		Implication between two HyperLTL formulas of the \ex fragment is decidable in exponential space.
 	\end{lemma}
 	
 	\begin{proof}
 		The proof is done analogously to the proof of Lemma~\ref{equi_forall}. \qed
 	\end{proof}
 	
 	

 	
 	Analogously to Theorem~\ref{theoremexall}, we obtain that checking implication between two alternation-free HyperLTL formulas is EXPSPACE-complete.
 	
 	\begin{theorem} 
 		Checking implication between alternation-free HyperLTL formulas is {EXPSPACE-complete}.
 	\end{theorem}
 	\begin{proof}
 		With Lemma~\ref{equi_forall} and Lemma~\ref{equi_exists}, the upper bound of Theorem~\ref{theoremexall} applies here as well. For the lower bound, we note that
 		the encoding in the proof of Theorem~\ref{theoremexall} is of the form
 		\[\exists \pi, \pi_{\mathit{zero}}, \pi_{\mathit{one}}.~ \forall \varset{\pi}'.~ \varphi_1(\pi) \wedge  \varphi_2(\pi_{\mathit{zero}}, \pi_{\mathit{one}}) \wedge \psi(\pi, \varset{\pi}'),\]
 		which is not an implication of alternation-free formulas. We can, however, transform this formula into an equisatisfiable formula by quantifying $\pi$ universally:
 		\[\exists \pi_{\mathit{zero}}, \pi_{\mathit{one}}.~ \forall \pi, \varset{\pi}'.~ \varphi_1(\pi) \wedge  \varphi_2(\pi_{\mathit{zero}}, \pi_{\mathit{one}}) \wedge \psi(\pi, \varset{\pi}')\]  
 		In the models of the new formula, the accepting computation of the Turing machine is simply represented on \emph{all} traces instead of on \emph{some} trace.
 		The formula is satisfiable iff the following implication between \ex formulas does \emph{not} hold:
 		\[
 		\exists \pi_{\mathit{zero}}, \pi_{\mathit{one}}.~\varphi_2(\pi_{\mathit{zero}}, \pi_{\mathit{one}})
 		\quad \mbox{ implies } \quad
 		\exists \pi, \varset{\pi}'.~ \neg (\varphi_1(\pi) \wedge \psi(\pi, \varset{\pi}'))
 		\]
 		Hence, we have reduced the problem whether an exponential-space bounded deterministic Turing machine accepts a certain input word to the implication problem between two \ex HyperLTL formulas. \qed
 	\end{proof}

 	
 	
 	With the results of this section, we have reached the borderline of the decidable HyperLTL fragments.
 	We will see in the next section that HyperLTL-SAT immediately becomes undecidable if the formulas contain a quantifier alternation that starts with a universal quantifier.
 	
 	\section{The Full Logic}
 	\label{allex}
 	\begin{figure}
 		\begin{align}
 		\forall \pi \exists \pi_s \exists \pi'.~&\bigg(\Big((\dot a, \dot a)_{\pi_s} \vee (\dot b, \dot b)_{\pi_s}\Big) \label{6.4} \\
 		&~~~~~~\wedge ((\tilde{a},\tilde{a})_{\pi_s} \vee (\tilde{b},\tilde{b})_{\pi_s}) \LTLuntil \LTLsquare (\#,\#)_{\pi_s}\bigg) \label{6.5} \\
 		&\wedge \LTLdiamond \LTLsquare (\#,\#)_{\pi} \label{6.6} \\
 		&\wedge \Bigg( \bigvee_{i\in\{1,2,3\}} \mathit{StoneEncoding}_i \label{6.9}\\
 		&~~~~~~\vee \LTLsquare (\#,\#)_\pi\Bigg) \label{6.31}
 		\end{align}
 		\caption{Reduction to HyperLTL for the PCP instance from Example~\ref{pcp_example}.}
 		\label{pcpToHyperEx}
 	\end{figure}
 	We now show that any extension beyond the already considered fragments makes the satisfiability problem undecidable. 
 	We prove this with a many-one-reduction from \emph{Post's correspondence problem} (PCP)~\cite{post1946variant} to the satisfiability of a $\forall\exists$ HyperLTL formula.
 	In PCP, we are given two lists $\alpha$ and $\beta$ consisting of finite words from some alphabet $\Sigma$.
 	For example, $\alpha$, with $\alpha_1 = a$, $\alpha_2 = ab$ and $\alpha_3 = bba$ and $\beta$, with $\beta_1 = baa$, $\beta_2 = aa$ and $\beta_3 = bb$, where $\alpha_i$ denotes the $i$\emph{th} element of the list, and $\alpha_{i_j}$ denotes the $j$\emph{th} symbol of the $i$\emph{th} element.
 	In this example, $\alpha_{3_1}$ corresponds to $b$.
 	PCP is the problem to find an index sequence $(i_k)_{1 \leq k \leq K}$ with $K \geq 1$ and $1 \leq i_k \leq n$ for all $k$, such that $\alpha_{i_1}\dots\alpha_{i_K}=\beta_{i_1}\dots\beta_{i_K}$.
 	We denote the finite words of a PCP solution with $i_\alpha$ and $i_\beta$ respectively.
 	
 	It is a useful intuition to think of the PCP instance as a set of $n$ domino stones. The first stone of our example is \stone{a}{baa}, the second is \stone{ab}{aa} and the third, and last, is \stone{bba}{bb}.
 	Those stones must be arranged (where copying is allowed) to construct the same word with the $\alpha$- and $\beta$-concatenations.
 	A possible solution for this PCP instance would be $(3,2,3,1)$, since the stone sequence \stone{bba}{bb} \stone{ab}{aa} \stone{bba}{bb} \stone{a}{baa} produces the same word, i.e., $bbaabbbaa = i_\alpha = i_\beta$.
 	For modelling the necessary correspondence between the $\alpha$ and $\beta$ components, we will use pairs of the PCP instance alphabet as atomic propositions, e.g., $(a,b)$.
 	We represent a stone as a sequence of such pairs, where the first position of the pair contains a symbol of the $\alpha$ component and the second position a symbol of the $\beta$ component.
 	For example, the first stone \stone{a}{baa} will be represented as $(a,b),(\#,b)(\#,a)$.
 	We will use $\#$ as a termination symbol.
 	Since the $\alpha$ and $\beta$ component of a stone may differ in its length, a sequence of stone representations might ``overlap''.
 	Therefore, we indicate the start of a new stone with a dotted symbol.
 	For example, we can string the first stone two times together: $(\dot{a},\dot{b}),(\dot{a},b)(\#,a)(\#,\dot{b})(\#,b)(\#,a)$.
 	In the following, we write $\tilde{a}$ if we do not care if this symbol is an $a$ or $\dot{a}$ and use $*$ as syntactic sugar for an arbitrary symbol of the alphabet.
 	We assume that only singletons are allowed as elements of the trace, which could be achieved by adding for every atomic proposition $(y_1,y_2)$ the conjunction $\bigwedge_{(y_1,y_2) \not = (y,y')} \LTLsquare (\neg ((y_1,y_2) \wedge (y,y')))$, for all $(y,y')$.
 	\begin{figure}
 		\begin{align}
 		&\mathit{StoneEncoding_3}=\\
 		&\bigg( \Big(((\dot b, \dot b)_\pi \wedge \LTLnext (b,b)_\pi \wedge \LTLnext \LTLnext (a,\dot *)_\pi \wedge \LTLnext \LTLnext \LTLnext (\dot *, \tilde{*})_\pi) \label{6.23} \\
 		&~~~~~~~~~~~~~~~~\vee ((\dot b, \dot b)_\pi \wedge \LTLnext (b,b)_\pi \wedge \LTLnext \LTLnext (a,\#)_\pi \wedge \LTLnext \LTLnext \LTLnext (\#, \#)_\pi) \Big) \label{6.24} \\
 		&~~~~~~~~~~~\wedge \LTLsquare (\LTLnext \LTLnext \LTLnext (\tilde{a}, *)_\pi \rightarrow (\tilde{a},*)_{\pi'}) \label{6.25} \\
 		&~~~~~~~~~~~\wedge \LTLsquare (\LTLnext \LTLnext \LTLnext (\tilde{b}, *)_\pi \rightarrow (\tilde{b},*)_{\pi'}) \label{6.26} \\
 		&~~~~~~~~~~~\wedge \LTLsquare (\LTLnext \LTLnext \LTLnext (\#, *)_\pi \rightarrow (\#,*)_{\pi'}) \label{6.27} \\
 		&~~~~~~~~~~~\wedge \LTLsquare (\LTLnext \LTLnext (*,\tilde{a})_{\pi} \rightarrow (*,\tilde{a})_{\pi'}) \label{6.28} \\
 		&~~~~~~~~~~~\wedge \LTLsquare (\LTLnext \LTLnext (*,\tilde{b})_{\pi} \rightarrow (*,\tilde{b})_{\pi'}) \label{6.29} \\
 		&~~~~~~~~~~~\wedge \LTLsquare (\LTLnext \LTLnext (*,\#)_{\pi} \rightarrow (*,\#)_{\pi'})\bigg) \label{6.30}
 		\end{align}
 		\caption{Formula in the reduction of the PCP instance from Example~\ref{pcp_example}, encoding that a trace may start with a valid stone 3 and that there must also exist a trace where stone 3 is deleted.}
 		\label{encoding3}
 	\end{figure}
 	\begin{example}
 		\label{pcp_example}
 		Consider, again, the following PCP instance with $\Sigma = \{a,b\}$.
 		Two lists $\alpha$, with $\alpha_1 = a$, $\alpha_2 = ab$ and $\alpha_3 = bba$ and $\beta$, with $\beta_1 = baa$, $\beta_2 = aa$ and $\beta_3 = bb$.
 		We can reduce this PCP instance to the question whether the HyperLTL formula shown in Figure~\ref{pcpToHyperEx} is satisfiable.
 		Let $\mathit{AP} \coloneqq (\{a,b,\dot a,\dot b\}\cup\{\#\})^{2}$.
 		The stone encoding is sketched with the example of stone $3$ in Figure~\ref{encoding3}.
 		
 		The subformula (\ref{6.4}) expresses that there exists a trace that starts with $(\dot a, \dot a)$ or $(\dot b, \dot b)$.
 		Intuitively, this means that there must exist a stone whose $\alpha$ and $\beta$ component start with the same symbol.
 		Subformula (\ref{6.5}) requires that there exists a ``solution'' trace $\pi_s$.
 		It ensures that the trace ends synchronously with $(\#,\#)^\omega$.
 		Combined, this guarantees that the word constructed from the $\alpha$ components is equal to the word constructed from the $\beta$ components, i.e., $i_\alpha = i_\beta$ for a PCP solution $i(k)$.
 		Subformula (\ref{6.6}) ensures that every trace eventually ends with the termination symbol $\#$. It is important to notice here that all traces besides $\pi_s$ are allowed to end asynchronously.
 		
 		It remains to ensure that trace $\pi_s$ only consists of valid stones. This is where the $\forall \exists$ structure of the quantifier prefix comes into play. The key idea is to use a $\forall\exists$ formula to specify that for every trace with at least one stone there is another trace \emph{with the first stone removed}.
 		Since we check that \emph{every} trace begins with a valid stone, this implies that all stones are valid. 
 		%
 		The encoding of stone 3 is exemplarily shown in Figure~\ref{encoding3}.
 		The first \emph{three} $\alpha$ components and the first \emph{two} $\beta$ components of the new trace are deleted.
 		The example set shown in Figure~\ref{pcpToHyperSet} shows this behavior for $\pi_s$, which starts with stone 3.
 		By deleting stone 3 from $\pi_s$ and shifting every position accordingly, we obtain $\pi'$.
 		Since $\pi'$  starts with a valid stone, namely stone 2, it satisfies subformula~(\ref{6.9}) for $i=2$. This requires that there exists another trace where stone 2 is deleted analogously. This argument is repeated until the trace is reduced to $(\#,\#)^\omega$, which is the only possibility for ``termination'' in the sense that $\pi_s$ ends synchronously with $(\#,\#)^\omega$.
 	\end{example}
 	Corresponding to this example, we can give a generalized reduction, establishing the undecidability of $\forall \exists$ formulas.
 	\begin{theorem}
 		\label{theorem-undec}
 		$\forall\exists$ HyperLTL-SAT is undecidable.
 	\end{theorem}
 	
 	\begin{proof}
 		Let a PCP instance with $\Sigma = \{a_1,a_2,..., a_n\}$ and two lists $\alpha$ and $\beta$ be given.
 		We choose our alphabet as follows: $\Sigma' = (\Sigma \cup \{\dot a_1, \dot a_2,... ,\dot a_n\} \cup {\#})^2$, where we use the dot symbol to encode that a stone starts at this position of the trace.
 		Again, we write $\tilde{a}$ if we do not care if this symbol is an $a$ or $\dot{a}$ and use $*$ as syntactic sugar for an arbitrary symbol of the alphabet.
 		We encode the idea from Example~\ref{pcp_example} in the following formula.
 		\[\varphi_{\text{reduc}} \coloneqq \forall \pi \exists \pi_s \exists \pi'.~ \varphi_{\text{sol}}(\pi_s) \wedge \varphi_{\text{validStone}}(\pi) \wedge \varphi_{\text{delete}}(\pi,\pi') \wedge \LTLdiamond \LTLsquare (\#,\#)_\pi\]
 		\begin{itemize}
 			\item
 			$\varphi_{\text{sol}}(\pi_s) \coloneqq (\bigvee_{i=1}^n(\dot a_i, \dot a_i)_{\pi_s}) \wedge (\bigvee_{i=1}^n (\tilde{a_i},\tilde{a_i})_{\pi_s}) \LTLuntil \LTLsquare (\#,\#)_{\pi_s}$
 			
 			We ensure that there exists a ``solution'' trace $\pi_s$, which starts pointed, i.e., where the $\alpha$ and $\beta$ components are the same. Accordingly to PCP, we require \emph{synchronous} ``termination''.
 			\item
 			$\varphi_{validStone}(\pi)$. This is ensured by a generalization of lines (\ref{6.23}) and (\ref{6.24}) of the stone encoding sketched in Figure~\ref{encoding3}.
 			
 			Every trace in the trace set starts with a valid stone. Note that we do \emph{not} require synchronous termination in any other trace than the ``solution'' trace.
 			
 			\item
 			$\varphi_{\text{delete}}(\pi,\pi')$. This is ensured by a generalization of lines (\ref{6.25}) to (\ref{6.30}) of the stone encoding sketched in Figure~\ref{encoding3}.
 			
 			By exploiting the $\forall \exists$ structure of the formula, we encode that for every trace $\pi$ there exists another trace $\pi'$ which is \emph{nearly} an exact copy of $\pi$ but with its first stone \emph{removed}.
 		\end{itemize}
 		\emph{Correctness}. We prove correctness of the reduction by showing that if there exists a solution, namely an index sequence $i(l)$ with $l \in \mathbb{N}$, for a PCP instance, then there exists a trace set $T$ satisfying the resulting formula $\varphi_{\text{reduc}}$ and vice versa. For the sake of readability, we, again, omit the set braces around atomic propositions, since we can assume that only singletons occur.
 		\begin{itemize}
 			\item
 			Assume there exists a solution $i$ to the given PCP instance with $|i_\alpha| = |i_\beta| = k$.
 			We can construct a trace set $T$ by building the trace $(i_\alpha[0],i_\beta[0])\ldots(i_\alpha[k],i_\beta[k])(\#,\#)^{\omega}$, denoted by $t_0$ and adding a dot to the symbol corresponding to the new stones start. We can infer the correct placement of the dots from the solution.
 			We distinguish two cases. If the solution is only of length $1$, we add $(\#,\#)^\omega$ to $T$ and successfully constructed a trace set satisfying the formula.
 			Otherwise, let $t_0$ start with stone $j$. We also add one of the following traces $t_1$ based on $t_0$ to $T$:
 			\begin{align*}
 			\text{if}~|\alpha_j| = |\beta_j| :~&(i[|\alpha_j|],i[|\beta_j|])\ldots(i[k],i[k])(\#,\#)^{\omega}\\ 
 			\text{if}~|\alpha_j| < |\beta_j| :~&(i[|\alpha_j|],i[|\beta_j|])\ldots(i[k],i[k-|\beta_j|+|\alpha_j|])\ldots(\#,i[k])(\#,\#)^{\omega}\\
 			\text{if}~|\alpha_j| > |\beta_j| :~&(i[|\alpha_j|],i[|\beta_j|])\ldots(i[k-|\alpha_j|+|\beta_j|],i[k])\ldots(i[k],\#)(\#,\#)^{\omega}
 			\end{align*}
 			We repeat adding traces $t_n$ based on the starting stone of every newly added trace $t_{n-1}$ until we terminate with $(\#,\#)^\omega$.
 			Note that $t_{n-1}$ might already end asynchronously.
 			By construction this is exactly a trace set $T$ satisfying $\varphi_{\text{reduc}}$.
 			\item
 			Let the formula $\varphi_{\text{reduc}}$ be satisfiable by a trace set $T$. Therefore, there exists a witness $t_0$ for $\pi_s$, which starts with a dot, whose $\alpha$ and $\beta$ components are the same at all positions, and which ends \emph{synchronously} with $(\#,\#)^\omega$.
 			$t_0$ also needs to start with a valid stone, which is ensured by the stone encoding, since otherwise $t_0 \not \in T$.
 			By construction there exists a subset $T_{\text{min}} \subseteq T$ that satisfies $\varphi_{\text{reduc}}$, which contains $t_0$ and every trace constructed by deleting one stone after another, with the last trace being $(\#,\#)^\omega$.
 			Because $t_0$ eventually terminates synchronously with $(\#,\#)$, the solution remains finite.
 			We define a total order for the traces in $T_{\text{min}}$ according to the number of dots or, equivalently, the number of stones. We also define a function $s$ that maps traces to the index of their starting stone. Let $A = [t_0,t_1,\ldots,t_n]$ be the list of traces in $T_{\text{min}}$ sorted in descending order. A possible solution for the PCP instance is the index sequence $s(t_0)~s(t_1)\ldots s(t_n)$.
 		\end{itemize}
 		Since we can use the construction from Section~\ref{secex}, the minimal undecidable fragment of HyperLTL is, in fact, $\forall \exists$. \qed
 	\end{proof}
 	
 	\begin{figure}[t]
 		\begin{align*}
 		&\emph{\text{Start}} \\
 		&\pi_s:~(\dot b, \dot b)(b,b)(a,\dot a)(\dot a, a)(b, \dot b)(\dot b, b)(b,\dot b)(a,a)(\dot a,a)(\#,\#) (\#,\#) \ldots \\
 		&\emph{\text{Delete stone 3}} \\
 		&\pi_s':~(\dot a, \dot a)(b,a)(\dot b, \dot b)(b,b)(a,\dot b)(\dot a, a)(\#,a)(\#,\#)(\#,\#) \ldots \\
 		&\emph{\text{Delete stone 2}} \\
 		&\pi_s'':~(\dot b, \dot b)(b,b)(a,\dot b)(\dot a,a)(\#,a)(\#,\#)(\#,\#) \ldots \\
 		&\emph{\text{Delete stone 3}} \\
 		&\pi_s''':~(\dot a, \dot b)(\#,a)(\#,a)(\#,\#)(\#,\#) \ldots \\
 		&\emph{\text{Delete stone 1}} \\
 		&\pi_s'''':~(\#,\#)(\#,\#) \ldots \\
 		&\emph{\text{End}}
 		\end{align*}
 		\caption{Trace set satisfying the formula from Example~\ref{pcp_example}, with omitted set braces around the atomic propositions.}
 		\label{pcpToHyperSet}
 	\end{figure}
 	
 	
 	\section{Conclusion}
 	\label{sumFut}
 	We have analyzed the decidability and complexity of the satisfiability problem for various fragments of HyperLTL. The largest decidable fragment of HyperLTL is the EXPSPACE-complete $\exists^*\forall^*$ fragment; the alternation-free \ex and \all formulas are PSPACE-complete; any fragment that contains the $\forall\exists$ formulas is undecidable. Despite the general undecidability, our
 	results provide a strong motivation to develop a practical SAT checker for HyperLTL. The key result is the PSPACE-completeness for the alternation-free fragment and the bounded \exall fragment, which means that for the important class of hyperproperties that can be expressed as a HyperLTL formula with a bounded number of exclusively universal or exclusively existential quantifiers, satisfiability and implication can be decided within the same complexity class as LTL. 
 	
 	There are several directions for future work. An important open question concerns the extension to branching time. HyperLTL is a sublogic of the branching-time temporal logic HyperCTL$^*$~\cite{DBLP:conf/post/ClarksonFKMRS14}. While the undecidability of HyperLTL implies that HyperCTL$^*$ is also, in general, undecidable (this was already established in~\cite{DBLP:conf/post/ClarksonFKMRS14}), the obvious question is whether it is possible to establish decidable fragments in a similar fashion as for HyperLTL.
 	
 	Another intriguing, and still unexplored, direction is the synthesis problem for HyperLTL (and HyperCTL$^*$) specifications. In synthesis, we ask for the existence of an implementation,
 	which is usually understood as an infinite tree that branches according to the possible inputs to a system and whose nodes are labeled with the outputs of the system. Since HyperLTL can express partial observability, the synthesis problem for HyperLTL naturally generalizes the well-studied synthesis  under incomplete information~\cite{kupfermant2000synthesis} and the synthesis of distributed systems~\cite{DBLP:conf/focs/PnueliR90}.
 	
 	Finally, it will be interesting to develop a practical implementation of the constructions presented in this paper and to use this implementation to analyze the relationships between various hyperproperties studied in the literature.
 	
 	
 	
 	\bibliography{newlib}

 \end{document}